 \documentclass[11pt]{amsart}

\usepackage{amsrefs}

\usepackage{epsfig}  		
\usepackage{epic,eepic}       
\usepackage{graphicx}

\usepackage{enumerate}
\usepackage{mathbbol}
\usepackage{amssymb}
\usepackage{braket}
\usepackage{longtable}
\usepackage[colorlinks]{hyperref}
\usepackage{doi}
\usepackage{caption}
\usepackage{mathtools}
\usepackage{microtype}
\usepackage{lineno}

\DeclareSymbolFontAlphabet{\mathbb}{AMSb}
\DeclareSymbolFontAlphabet{\mathbbl}{bbold}

\newtheorem{lemma}{Lemma}[section]
\newtheorem*{lemma*}{Lemma}
\newtheorem{theorem}[lemma]{Theorem}
\newtheorem*{theorem*}{Theorem}

\newtheorem{proposition}[lemma]{Proposition}

\newtheorem*{proposition*}{Proposition}
\newtheorem{fact}[lemma]{Fact}
\newtheorem*{fact*}{Fact}

\newtheorem*{notation*}{Notation}
\newtheorem*{conventions*}{Conventions}

\newtheorem*{remark*}{Remark}

\newtheorem*{corollary*}{Corollary}

\newtheorem*{conjecture*}{Conjecture}

\newtheorem*{problem*}{Problem}

\newtheorem*{question*}{Question}

\newtheorem{assumption*}{Assumption}

\theoremstyle{definition}

\newtheorem*{example*}{Example}
\newtheorem{definition}[lemma]{Definition}
\newtheorem*{definition*}{Definition}

\theoremstyle{remark}

\newtheorem*{claim*}{Claim}

\newtheorem*{case*}{Case}
\newtheorem*{construction*}{Construction}

\newtheorem*{exercise*}{Exercise}

\numberwithin{equation}{section}

\newcommand{\N}{\mathbb{N}}
\newcommand{\Z}{\mathbb{Z}}

\newcommand{\bs}{\backslash}

\newcommand{\proves}{\vdash}

\newcommand\CC{{\mathcal C}}
\newcommand\DD{{\mathcal D}}

\newcommand\FF{{\mathcal F}}
\newcommand\GG{{\mathcal G}}
\newcommand\HH{{\mathcal H}}

\newcommand\LL{{\mathcal L}}

\newcommand\<{\langle}
\renewcommand\>{\rangle}

\def\Ind#1#2{#1\setbox0=\hbox{$#1x$}\kern\wd0\hbox to 0pt{\hss$#1\mid$\hss}
\lower.9\ht0\hbox to 0pt{\hss$#1\smile$\hss}\kern\wd0}

\def\notind#1#2{#1\setbox0=\hbox{$#1x$}\kern\wd0
\hbox to 0pt{\mathchardef\nn=12854\hss$#1\nn$\kern1.4\wd0\hss}
\hbox to 0pt{\hss$#1\mid$\hss}\lower.9\ht0 \hbox to 0pt{\hss$#1\smile$\hss}\kern\wd0}

\def\includeE#1{{\lhook\kern-3.5pt\joinrel\smash{
    \mathop{\longrightarrow}\limits^{#1}}}}

\def\efor/{Example~\ref{E4}}

\def\BL/{Baldwin--Lachlan}
\def\Bu/{Buechler}
\def\Hr/{Hrushovski}
\def\lm/{locally modular}
\def\wm/{weakly minimal}
\def\nm/{non--modular}
\def\ss/{superstable}
\def\ud/{unidimensional}
\def\sm/{strongly minimal}

\def\hbar{\bar{h}}

\def\xbar{\bar{x}}

\def\tr/{trivial}
\def\nt/{non--trivial}
\def\st/{strong type}

\def \F{{\mathcal F}}
\def\Z{{\mathbb Z}}

\def\Fa0{{\FF^a_{\aleph_0}}}

\def\<{\langle}
\def\>{\rangle}

\title{Model checking in finite fields and finite groups}
\author{Samuel Braunfeld}
\thanks{Supported by Project 24-12591M of the Czech Science Foundation (GA\v{C}R), and supported partly by the long-term strategic development financing of the Institute of Computer Science (RVO: 67985807).
}

\renewcommand{\F}{\mathbb{F}}

\begin{document}
	\begin{abstract}
		We prove the following results.
		
		\begin{enumerate}
			\item First order model checking is fixed-parameter tractable on the class of finite fields, as a corollary of results of Ax on the theory of (pseudo)finite fields.
			\item Every hereditary graph class first order definable in the class of finite groups is monadically stable, and thus has fixed-parameter tractable first order model checking.
			\item Monadic second order model checking is not slicewise polynomial on the class of cyclic groups of prime-power order, unless E = NE.
		\end{enumerate}
	\end{abstract}
\maketitle

\section{Introduction}

There has been a long line of results concerning the fixed-parameter tractability of model checking for first order logic or monadic second order logic (MSO) in hereditary graph classes. This is the algorithmic problem where one is given a structure $M$ and a sentence $\phi$ of the logic, and determines if $M \models \phi$. It is \emph{fixed-parameter tractable} on a class $\CC$ of structures if there is an algorithm that runs in time $f(|\phi|)|M|^c$ for some computable function $f$ and some constant $c$ (that is independent of $\phi$), for every $M \in \CC$. This work has brought together structural graph theory, model theory, automata theory, and computer science. In the MSO setting, it includes Courcelle's theorem showing tractability for classes of bounded treewidth \cite{courcelle1990monadic}, and its extension to classes of bounded cliquewidth \cite{courcelle2000linear}. In the first order setting, it includes Seese's theorem showing tractability for classes of bounded degree \cite{seese1996linear}, several results for sparse graph classes culminating in tractability for nowhere dense classes \cite{grohe2017deciding}, and ongoing work in generalized sparsity that notably incorporates ideas from model-theoretic (neo)stability theory \cite{siebertz2024advances}. 

In a survey of this work twenty years ago, Grohe raised the question of the fixed-parameter tractability of model checking for various classes of finite algebraic structures \cite[Open Problem 8.2]{grohe2007logic}. The only work towards this problem we are aware of are \cite{abu2016algorithmic} and \cite{bova2015first}. In \cite{abu2016algorithmic}, first order tractability is shown for the the classes of finite boolean algebras and finite abelian groups by interpreting them in structures that admit automata-theoretic analysis. In \cite{bova2015first}, first order tractability is shown for abelian groups by using the fact that modules admit quantifier elimination down to existential formulas, and MSO intractability (in particular, W[1]-hardness) is shown for a subclass of abelian groups presented as inputs doubly logarithmic in the size of the group. 

In this work, we prove first order tractability for the class of finite fields.

\begin{theorem}[Theorem \ref{thm:fomod}]
	The class of finite fields has fixed-parameter tractable first order model checking. Furthermore, it can be done in time $f(|\phi|)k^3$ for some computable function $f$ if the field is given just as a binary string of length $k$ specifying the field's cardinality.
\end{theorem}

This is obtained as an immediate corollary of a result of Ax \cite{ax1968elementary} showing that finite fields admit a sort of computable asymptotic quantifier elimination down to boolean combinations of existentials, similar to the approach of \cite{bova2015first} for abelian groups. 

Under a standard complexity assumption, we also prove MSO intractability for the class of finite cyclic groups of prime power order.

\begin{theorem}[Theorem \ref{thm:msocyc}]
		Assume E $\neq$ NE. Then, on the class of cyclic groups of prime-power order, there is an MSO sentence $\phi$ such that checking whether $G \models \phi$ cannot be done in time polynomial in $|G|$, and so MSO model checking is not in XP.
\end{theorem}

This is proved by defining a grid and simulating a Turing machine for a unary language in NP$\bs$P.

The other result of the paper (Proposition \ref{prop:fogroup}) is that a natural method for trying to show first order intractability for the class of finite groups by suitably encoding the class of all graphs, or any intractable hereditary graph class, cannot work, by a simple counting argument. 

\subsection{Acknowledgments}
We thank Emil Je\v{r}\'abek for pointing out a simple MSO axiomatization of finite fields, and Stefan Ludwig for pointing out errors in a previous version.

\section{First order model checking}

In this section, we assume that all sentences are in first order logic. Fields are equipped with the language $\LL_{ring} = \set{+, \times, 0, 1}$ while groups are equipped with the language $\LL_{group} = \set{\cdot, e}$.

\subsection{Finite fields}

In this subsection, we show first order tractability for the class of finite fields, as a corollary of the following result of Ax. First order tractability for the class $\set{\Z/n\Z | n \in \N}$ then follows by an application of Feferman-Vaught. All sentences are assumed to be in $\LL_{ring}$.

\begin{definition}
	A \emph{reducibility sentence} is a boolean combination of sentences of the form $\exists x \phi(x)$ where $x$ is a singleton and $\phi$ is quantifier-free.
\end{definition}

\begin{fact}[{\cite[Theorem 12]{ax1968elementary}}] \label{fact:ax12}
	Let $\phi$ be a sentence. Then there exist a reducibility sentence $\psi$ and some $N_\phi \in \N$ such that for every prime power $n \geq N_\phi$, we have $\F_n \models \phi \leftrightarrow \psi$. Furthermore, both $\psi$ and $N_\phi$ are computable from $\phi$.
\end{fact}

\begin{theorem} \label{thm:fomod}
	The class of finite fields has fixed-parameter tractable first order model checking. Furthermore, it can be done in time $f(|\phi|)k^3$ for some computable function $f$ if the field is given just as a binary string of length $k$ specifying the field's cardinality.
\end{theorem}
\begin{proof}
	Given a sentence $\phi$, we compute a reducibility sentence $\psi$ and an $N_\phi \in \N$ as given by Fact \ref{fact:ax12}. Now suppose we are given a field $\F_n$. We may suppose that $n \geq N_\phi$ and $n \geq |\psi|$, and so $n$ is larger than any coefficient or degree of any polynomial appearing in $\psi$, since otherwise we may brute-force check whether $\F_n \models \phi$. We can then check in time poly($n$) whether each polynomial appearing in $\psi$ has a root in $\F_n$. By, for example \cite[Theorem 20.1]{shoup2009computational}, we may determine if a monic polynomial in $\F_n[X]$ of degree at most $d$ has a root in $\F_n$ using $O(d^3 \log(n))$ field operations. By our assumption on the size of $n$, we may reduce the coefficients of every polynomial in $\psi$ to lie in the prime field of $\F_n$ in polynomial time, and so the fixed-parameter tractability result follows.
	
	For the finer result, as noted after \cite[Theorem 20.1]{shoup2009computational}, we may perform each field operation in time $O(\log^2(n))$, to obtain a total time of $O(\log^3(n))$.
\end{proof}

We briefly remark on the proof of Fact \ref{fact:ax12}, since it relevant for the complexity of the function $f$ in Theorem \ref{thm:fomod}. The proof goes through the theory $T_{pf}$ of pseudofinite fields, i.e. infinite models of the common theory of all finite fields, which \cite{ax1968elementary} shows is complete and recursively axiomatizable. Given $\phi$, it is first shown that there is a suitable $\psi$ such that $T_{pf} \proves \phi \leftrightarrow \psi$. Then the way that $\psi$ is computed is by a brute-force search over all proofs of the form $T_{pf} \proves \phi \leftrightarrow \theta$ for some reducibility sentence $\theta$, using that $T_{pf}$ is recursively axiomatizable. Thus we have no better upper bound on the function $f$ in Theorem \ref{thm:fomod} than that it is computable. However, it seems likely that the results of \cite{fried1976solving} could be used to show that $f$ can be taken primitive recursive.

In the case of hereditary classes of relational structures, it is conjectured that first-order model checking is fixed-parameter tractable on a class $\CC$ if and only if $\CC$ has the model-theoretic property monadic NIP, which coincides with the weaker property NIP in hereditary classes of relational structures \cite{braunfeld2022existential}. In \cite[Theorem 1.7]{dreier2024flip}, the implication that the failure of (monadic) NIP implies that first-order model checking is AW[$*$]-hard for hereditary graph classes is proved. Theorem \ref{thm:fomod} shows that the failure of NIP does not in imply the intractability of first-order model checking outside of the hereditary relational setting, since $T_{pf}$, and thus the class of finite fields, does not have NIP. For an example using a non-hereditary graph class, one can take the class of Paley graphs (see \cite{bollobas1981graphs}), which are definable in finite fields but whose limit theory is the Rado graph. (Tractability for Paley graphs can also be proved more directly by computing which extension axioms are needed to prove a given sentence or its negation, and then using the bound in \cite[Theorem 3]{bollobas1981graphs} on how large Paley graphs need to be for the extension axioms to hold.)

\subsection{Finite groups}

If a hereditary graph class $\CC$ is not monadically NIP (which conjecturally coincides with intractable first order model checking), then this is witnessed by the fact that the class of all graphs is definable in $\CC$ with only polynomial blow-up \cite[Theorem 1.8]{dreier2024flip}. Using a counting argument, we show there cannot be an analogous result giving intractability for the class of finite groups. More generally, any hereditary graph class definable in the class of finite groups with only polynomial blow-up must be monadically stable, and thus have tractable first order model checking \cite{dreier2024first}. We refer to \cite[\S 9.1]{siebertz2024advances} for the definitions of monadic stability and monadic NIP.

\begin{proposition} \label{prop:fogroup}
	Let $\CC$ be a hereditary graph class. If $\CC$ is definable in the class of finite groups and there is a polynomial $p(n)$ such that each graph $G$ is definable in a group of size at most $p(|G|)$, then $\CC$ is monadically stable.
\end{proposition}
\begin{proof}
	Given a class of structures $\DD$, let $\DD_n$ denote the number of structures in $\DD$ of size at most $n$, up to isomorphism. Let $\GG$ be the class of finite groups, and $\CC$ be a hereditary class of finite graphs definable in the class of finite groups, with a polynomial $p(n)$ as in the statement. By \cite[Theorem 1]{mciver1987enumerating}, $\GG_n \leq n^{O(\log^2(n))}$. If the definition of $\CC$ uses a formula $Dom(\xbar)$ for the domain of the graph, then $\CC_n \leq \GG_{p(n)^{|\xbar|}} \leq n^{O(\log^2(n))}$. By \cite[Theorem 1.3]{braunfeld2022existential}, $\CC$ must be monadically NIP. Letting $\HH$ be the hereditary closure of the class of half-graphs, we have $\HH_n \geq 2^{n-2}$ (e.g., see \cite[Examples 14(7)]{pouzet2006profile}). So $\CC$ must also be edge-stable, and thus monadically stable by \cite[Theorem 1.3]{nevsetvril2021rankwidth}.
\end{proof}

\section{MSO model checking in cyclic groups of prime power order}
In this section, we assume all sentences are in MSO. Because the groups we analyze are cyclic, we will use additive notation and equip them with the language $\LL_{cyc} = \set{+, 0}$. We show the MSO intractability of cyclic groups of prime power order.

\begin{definition}
The complexity class E (resp. NE) is the class of decision problems that can be solved by a deterministic (resp. non-deterministic) Turing machine in time $2^{O(n)}$.
\end{definition}

\begin{definition}
	Let $\CC$ be a class of structures, and let $\LL$ be a logic. We say that $\LL$-model checking is in XP (or is \emph{slicewise polynomial}) if there is an algorithm that runs in time $|M|^{O(f(|\phi|))}$ for some computable function $f$, for every $M \in \CC$.
\end{definition}

\begin{lemma} \label{lem:unary_prime}
	Let $L \subset \set{1}^*$ be a unary language. There is a unary language $L'$ containing only strings of prime length such that $L$ is polynomial-time equivalent to $L'$.
\end{lemma}
\begin{proof}
	Let $p_n$ denote the $n^{th}$ prime. By the prime number theorem, $p_n = (1+o(1))n\ln n$. Given $L$, let $L' = \set{1^{p_n} | 1^n \in L}$. 
	
	We first reduce $L$ to $L'$. To decide if $1^n \in L$, we may compute $p_n$ in polynomial time by iterating through $\N$ and running a primality test on each number, and then check whether $1^{p_n} \in L'$. 
	
	For the other direction, to decide if $1^n \in L'$, we compute $k$ such that $n = p_k$ (rejecting if there is no such $k$), again by iterating through $\N$ and running a primality test on each number. Then we check whether $1^k \in L$.
\end{proof}

\begin{theorem} \label{thm:msocyc}
	Assume E $\neq$ NE. Then MSO model checking on the class of cyclic groups of prime-power order is not in XP.
\end{theorem}
\begin{proof}
	By \cite[Theorem 1]{book1974tally}, E $\neq$ NE is equivalent to the existence of a unary language $L \subset \set{1}^*$ in NP that is not in P. By Lemma \ref{lem:unary_prime}, we may assume $L$ contains only strings of prime length. Fix a non-deterministic Turing machine $M$ for $L$ and a $k \in \N$ such that $M$ runs in time $o(n^k)$. To prove the statement, it suffices to produce an $\LL_{cyc}$-sentence $\theta$ such that $C_{p^{2k}} \models \theta$ if and only if $1^p \in L$, for $p$ sufficiently large so that $p^k$ is larger than runtime of $M$ on $1^p$. The idea is that $\theta$ will define a $p^k \times p^k$ grid, and by quantifying over set-variables corresponding to an assignment of either 1 or blank to each vertex, will check whether some assignment corresponds to an accepting run of $M$ on input $1^p$.
	
	We begin by defining a $p^k \times p^k$ grid in $C_{p^{2k}}$. First, note that for each $i \in [2^k]$, the subgroup of order $p^i$ is definable, by a formula saying $X$ is a subgroup with exactly $i$ proper subgroups, so we may assume we have a predicate $C_{p^i}$ for each subgroup. We first write a formula $Gen_{p^k}(x)$ defining that $x$ is a generator for $C_{p^k}$, by saying that $x \in C_{p^k}$ and such that there is no proper subgroup of $C_{p^k}$ containing $x$. We next give a formula $Rep_{p^k}(x, X)$ defining another set $T$ and ``generator'' $t$ such that $\set{0, t, 2t, \dots, (p^k-1)t}$ is a transversal of the cosets of $C_{p^k}$.
		Let $Rep_{p^k}(x, X)$ say that no proper subgroup contains $x$ and that $X$ is a minimal set such that $0 \in X$ and if $x' \in X$ then $x' + x \in X \cup C_{p^k}$. Suppose $h$, $t$, and $T$ are such that $C_{p^{2k}} \models Gen_{p^k}(h) \wedge Rep_{p^k}(t, T)$. Then for every $g \in C_{p^{2k}}$, there are unique $x \in C_{p^k}, y \in T$ such that $g = x+y$, so we can identify $C_{p^{2k}}$ with $C_{p^k} \times T$. To define the (directed) edges of the grid, we let $y$ be the right-neighbor of $x$ if $y \not\in T$ and $y = x+h$, and we let $y$ be the up-neighbor of $x$ if $y \not\in C_{p^k}$ and $y = x+t$. So each row of the grid is identified with a coset of $C_{p^k}$, with the bottom row identified with $C_{p^k}$ itself, while each column gives a transversal of the cosets with the left column identified with $T$. We also define the set $U_{in} = \set{0, h, \dots, (p-1)h}$ consisting of the first $p$ elements of the bottom row of the grid by the formula saying that $U_{in}$ is the minimal set containing $0$ and such that if $x \in U_{in}$ then $x+h \in U_{in} \cup C_{p^{k-1}}$.

After existentially quantifying over the existence of $h$, $t$, and $T$ such that $C_{p^{2k}} \models Gen_{p^k}(h) \wedge Rep_{p^k}(t, T)$, it is standard that we may write a sentence $\theta$ involving a unary predicate $U_1$ so that $C_{p^{2k}} \models \theta$ if and only if, interpreting each row of the grid as the tape of a Turing machine, each vertical step in the grid as a time step of the machine, and $U_1$ as the 1-cells of the tapes, the expansions of the grid by $U_1$ encodes an accepting run of the machine $M$ (e.g., see \cite[Lemma 4.5]{kreutzer2012parameterized}); by also requiring that $U_1 \cap C_{p^k} = U_{in}$, we enforce that the input tape contains $1^p$. Since we have included enough space in the grid for an accepting computation, we will have $C_{p^{2k}} \models \theta$ if and only if $1^p \in L$.
\end{proof}

As noted after Theorem \ref{thm:fomod}, pseudofinite fields, i.e. infinite models of the common first order theory of finite fields, are the key notion for that result. Thus it seems natural to consider the MSO analogue of this notion. However, the class of finite fields is finitely axiomatizable in MSO via the field axioms and the fact that the multiplicative group is cyclic, and so there are no pseudofinite fields in this setting.

 \bibliographystyle{alpha}
\bibliography{Bib}

\end{document}